\documentclass[10pt]{iopart}
\pdfoutput=1

\usepackage{iopams}

\expandafter\let\csname equation*\endcsname\relax
\expandafter\let\csname endequation*\endcsname\relax
\usepackage[english]{babel} \usepackage{amsmath} \usepackage{color}
\usepackage{epsfig} \usepackage{graphicx} \usepackage{bm}
\usepackage{mathtools}
\usepackage{amsthm}
\usepackage{times,amsmath,amssymb} \usepackage{amsfonts}
\usepackage{amssymb}

\usepackage[colorlinks,urlcolor=blue,bookmarks=false,hypertexnames=true]{hyperref}

\allowdisplaybreaks 

\newcommand{\media}[1]{\left\langle #1 \right\rangle}

\newcommand{\ket}[1]{| #1 \rangle}
\newcommand{\scp}[2]{\langle #1 | #2 \rangle}
\newcommand{\braket}[3]{\langle #1 | #2 | #3 \rangle}

\newcommand{\Span}{\mathrm{span}} %
\newcommand{\norm}[1]{\left\lVert#1\right\rVert}
\newtheorem{theorem}{Theorem}
\newtheorem{lemma}[theorem]{Lemma} 
\newtheorem{corollary}[theorem]{Corollary}
\newtheorem{definition}{Definition}

\DeclareMathOperator{\TDlim}{TD-lim}

\definecolor{myred}{RGB}{168,5,14}
\definecolor{myblue}{RGB}{13,13,255}
\definecolor{editorcolor}{RGB}{168,5,14}
\definecolor{mygreen}{RGB}{20,150,20}

\begin{document}

\title[Rigorous existence and location of QPTs in lattice Hamiltonian
systems]{Rigorous existence and location of quantum phase transitions
  in lattice Hamiltonian systems}

\author{Massimo Ostilli} \address{Instituto de F\'isica, Universidade
  Federal da Bahia, Salvador 40170-115, Brazil}

\author{Carlo Presilla} \address{Dipartimento di Matematica, Sapienza
  Universit\`a di Roma, Piazzale A. Moro 2, Roma 00185, Italy}
\address{Istituto Nazionale di Fisica Nucleare, Sezione di Roma 1,
  Roma 00185, Italy}

\vspace{10pt}
\begin{indented}
\item[]\today
\end{indented}

\begin{abstract}
  We extend the analysis of the class of quantum phase transitions
  (QPTs) that can be interpreted as condensations in state space,
  first introduced in [M. Ostilli and C. Presilla, J. Phys. A
  \textbf{54}, 055005 (2021)], by generalizing the arguments of
  [M. Ostilli and C. Presilla, Phys. Rev. Lett. \textbf{127}, 040601
  (2021)] to prove the existence and determine the location (via
  simple bounds) of QPTs in general one-parameter lattice
  Hamiltonians. Unlike our original formulation, this extension also
  encompasses second-order QPTs, for which we provide the explicit
  example of the transverse-field Ising model.  Our analysis suggests
  that, under conditions typically satisfied in physical contexts, any
  QPT taking place in lattice systems can be interpreted as a
  condensation in state space.
\end{abstract}


\section{Introduction}
Quantum phase transitions (QPTs) play a central role in many areas of
modern physics. Over the past century since the advent of quantum
mechanics, they have been investigated through a variety of
approaches, including exact solutions, perturbation theory,
renormalization-group methods, fidelity-based approaches, and
numerical simulations~\cite{SGCS,KB,Vojta,Sachdev,Fidelity,Carr}. A
common element underlying these methods is the central idea ---
originally formulated by Lee and Yang for classical
systems~\cite{LeeYang} --- that a phase transition is associated with
the breakdown of analiticity in the free energy of the system. At zero
temperature, this corresponds to a singular behavior in the
ground-state (GS) energy in the thermodynamic limit. More precisely,
when the system size tends to infinity while the particle density is
kept constant, some derivative of the GS energy per particle becomes
noncontinuous at the critical point and may even diverge. If the first
derivative exhibits a finite discontinuity, the transition is said to
be first order; otherwise, if the discontinuity appears only in higher
derivatives, the transition is classified as second order.

In recent years, we have proposed a new approach for Hamiltonian
lattice systems in which the critical point --- if present --- emerges
from the crossing of two analytic functions, called $E_\mathrm{cond}$
and $E_\mathrm{norm}$, representing the ground-state (GS) energy of
the Hamiltonian restricted to two suitable subspaces of the original
Hilbert space. One of these subspaces constitutes an infinitesimal
fraction of the full state space in the thermodynamic limit. The
underlying mechanism, first introduced in Ref.~\cite{QPTA}, allows a
QPT to be interpreted as a condensation in state space. Afterwards, we
extended and applied this approach to several
models~\cite{WC_QPT,FTQPT,FTQPT2}, all of which exhibit first-order
QPTs. In particular, in Ref.~\cite{WC_QPT} we showed that the Wigner
crystal~\cite{Wigner} forms through a first-order QPT.  In general,
calculating analytically the functions $E_\mathrm{cond}$ and
$E_\mathrm{norm}$ represents a formidable task, an issue that might
render our theoretical approach not pragmatic. More precisely, the
numerical complexity for determining $E_\mathrm{cond}$ and
$E_\mathrm{norm}$ is not different from the numerical complexity for
determining the full GS energy of the system. However, the special
arguments that we introduced in Ref.~\cite{WC_QPT} change this
perspective radically.  Here, on generalizing these analytic
arguments, we prove the existence and determine the location of QPTs
in a general one-parameter lattice Hamiltonian. In particular, unlike
our original formulation, this extension also encompasses second-order
QPTs, for which we provide the explicit example of the
transverse-field Ising model. Contrary to our original claims, our
approach appears to apply to any type of QPT; that is, the possibility
of interpreting a QPT as a condensation in state space seems to hold
in full generality. As we shall show, the concavity of the GS energy
plays a central role in our general theorem and leads to
straightforward formulas for bounding the critical point.

\section{QPTs as condensations in the state space}
In this section, we recall the mechanism leading to QPTs firstly
introduced in~\cite{QPTA} and then partially extended
in~\cite{WC_QPT}.  Let us consider a lattice model with $N$ sites and
$N_\mathrm{p}$ particles or spins described by a Hamiltonian
\begin{align}
  \label{H}
  H=K+g V,
\end{align}
where $K$ and $V$ are Hermitian non-commuting operators, and $g\geq 0$
a free dimensionless parameter. When dealing with systems of spins
located at fixed lattice positions, it will be understood that
$N_\mathrm{p}=N$. Regardless of the details of $K$ and $V$, we
represent $H$ in the eigenbasis of $V$ and it is natural to call $V$
the potential operator, and $K$ the hopping operator, for which we
assume null diagonal matrix elements in the eigenbasis of $V$.  To
exclude trivial behaviors, we suppose that the eigenvalues of $K$ and
$V$ scale linearly with the number of particles $N_\mathrm{p}$.  Since
in the two opposite limits $g\to 0$ and $g\to \infty$, the GS of the
system tends to the GS of $K$ and $V$, respectively, we wonder if, in
the thermodynamic limit, or $\TDlim$ for brevity, see later for a
precise definition, this transition occurs as a QPT taking place at
some critical value $g_\mathrm{c}$, whose location we also want to
establish, at least approximately.

In proposing to view any QPT in lattice Hamiltonian systems as a
condensation in the space of states, we approach the above issues as
follows.  We decompose the Hilbert space $\mathcal{H}$ of the system
as the direct sum of two mutually orthogonal subspaces, denoted
condensed and normal, namely,
\begin{align}
  \label{Space}
  \mathcal{H}=\mathcal{H}_\mathrm{cond} \oplus \mathcal{H}_\mathrm{norm}.
\end{align}
The definition of these subspaces is as follows.  We write
$\mathcal{H} = \Span \{ \ket{n} \}_{n=1}^{M}$, where $\{ \ket{n} \}$
(later on called configurations) is a complete orthonormal set of
eigenstates of $V$, i.e., we have $V \ket{n} =V_n \ket{n}$,
$n=1,\dots,M$, where we assume ordered, possibly degenerate, potential
(eigen-)values $V_1 \leq V_2 \leq \dots \leq V_M$. Given an integer
$M_\mathrm{cond}<M$, we then define
$\mathcal{H}_\mathrm{cond} = \Span \{ \ket{n}
\}_{n=1}^{M_\mathrm{cond}}$ and
$\mathcal{H}_\mathrm{norm} = \Span \{ \ket{n}
\}_{n=M_\mathrm{cond}+1}^{M} = \mathcal{H}_\mathrm{cond}^\perp$.  This
definition essentially relies on the choice of the dimension
$M_\mathrm{cond}$, which, in view of the ordering of the potential
values, marks the maximum potential value included in the condensed
subspace
\begin{align}
  \label{maxVcond}
  \max V_\mathrm{cond} =
  \max\{V_n:~\ket{n}\in\mathcal{H}_\mathrm{cond}\} =
  V_{M_\mathrm{cond}}.
\end{align}
Consider the GS energy of the Hamiltonian $H$ in the full Hilbert
space $\mathcal{H}$, alongside the GS energies of $H$ restricted to
the condensed and normal subspaces:
\begin{align}
  E &=\inf_{\ket{u}\in\mathcal{H}} \braket{u}{H}{u}/\scp{u}{u},
  \\
  \label{Econd}
  E_\mathrm{cond} &= \inf_{\ket{u}\in\mathcal{H}_\mathrm{cond}}
                    \braket{u}{H}{u}/\scp{u}{u},
  \\
  E_\mathrm{norm} &= \inf_{\ket{u}\in\mathcal{H}_\mathrm{norm}}
                    \braket{u}{H}{u}/\scp{u}{u}.
\end{align}
We are interested in the situations where $M_\mathrm{cond}/M\ll 1$
and, as a consequence, $M_\mathrm{norm}/M \simeq 1$, where
$M_\mathrm{norm} = (M-M_\mathrm{cond})$.  This justifies the names
\textit{condensed} and \textit{normal} assigned to the two subspaces
and suggests the following dichotomy argument: since
$\mathcal{H}\simeq\mathcal{H}_\mathrm{norm}$, we have
$E\simeq E_\mathrm{norm}$ --- unless --- it is energetically more
convenient to ``freeze'' the system into the infinitely (in the
TD-lim) smaller subspace $\mathcal{H}_\mathrm{cond}$, where we get
$E\simeq E_\mathrm{cond}$.

The above heuristic argument can be cast in rigorous terms as follows.
The $\TDlim$ is defined as the limit $N,N_\mathrm{p}\to\infty$ with
$N_\mathrm{p}/N=\varrho$ constant.  Consider the rescaled per particle
energies
\begin{definition}
  \label{RescaledEnergies}
  \begin{align}
    \epsilon(g) &= \TDlim E(N,N_\mathrm{p},g)/N_\mathrm{p},
    \\
    \epsilon_\mathrm{cond}(g)
                &=\TDlim E_\mathrm{cond}(N,N_\mathrm{p},g)/N_\mathrm{p},
    \\
    \epsilon_\mathrm{norm}(g)
                &=\TDlim E_\mathrm{norm}(N,N_\mathrm{p},g)/N_\mathrm{p},
  \end{align}
\end{definition}\noindent
which are finite in view of the assumed scaling properties of $K$ and
$V$ (dependence on $\varrho$ is left understood).  Given $N$,
$N_\mathrm{p}$, and the Hermitian operators $K$ and $V$, we shall
indicate the minimum eigenvalue of $K$ by $\min K$; the minimum
eigenvalue of $V$ as $\min V$ ($=V_1$); the maximum eigenvalue of $V$
as $\max V$ ($=V_M$); and the mean classical value of $V$ as
$\overline{V}=\sum_{n=1}^{M} V_n/M$~\cite{Note}.  These definitions
have their natural counterparts for the operators $K$ and $V$
restricted to the subspaces $\mathcal{H}_{\mathrm{cond}}$ and
$\mathcal{H}_{\mathrm{norm}}$.  For example, by $\min K_\mathrm{cond}$
we mean the smallest eigenvalue of the operator $K$ restricted to
$\mathcal{H}_{\mathrm{cond}}$, and so on.  We shall also make use of
the following per particle TD-lim related quantities
\begin{definition}
  \label{Notation}
  \begin{align}
    &\media{v}=\TDlim \frac{\media{V}}{N_\mathrm{p}}
      =\TDlim \frac{1}{N_\mathrm{p}}
      \left(\frac{1}{M} \sum_{n=1}^{M} V_n \right),
      \quad \mathrm{(classical~mean ~value~ of~} V),\\
    &\min v=\TDlim \frac{\min V}{N_\mathrm{p}},
      \quad \mathrm{(minimal~ eigenvalue~of~} V),\\
    &\max v=\TDlim \frac{\max V}{N_\mathrm{p}},
      \quad \mathrm{(maximal~ eigenvalue~of~} V),\\
    \label{maxvcond1}
    &\max v_\mathrm{cond}=\TDlim \frac{\max V_\mathrm{cond}}{N_\mathrm{p}},
      \quad \mathrm{(maximal~ eigenvalue~of~}
      V \vert_{\mathcal{H}_\mathrm{cond}}),
    \\
    &\min v_\mathrm{norm}=\TDlim \frac{\min V_\mathrm{norm}}{N_\mathrm{p}},
      \quad (\mathrm{minimal~ eigenvalue~of~}
      V \vert_{\mathcal{H}_\mathrm{norm}}) ,
    \\
    &\min k=\TDlim \frac{\min K}{N_\mathrm{p}}
      =\epsilon(0),
      \quad \mathrm{(minimal~ eigenvalue~of~} K),\\
    &\min k_\mathrm{cond}=\TDlim \frac{\min K_\mathrm{cond}}{N_\mathrm{p}}
      =\epsilon_\mathrm{cond}(0),
      \quad \mathrm{(minimal~ eigenvalue~of~}
      K \vert_{\mathcal{H}_\mathrm{cond}}),\\
    &\min k_\mathrm{norm}=\TDlim \frac{\min K_\mathrm{norm}}{N_\mathrm{p}}
      =\epsilon_\mathrm{norm}(0),
      \quad \mathrm{(minimal~ eigenvalue~of~}
      K \vert_{\mathcal{H}_\mathrm{norm}}).   
  \end{align}
\end{definition}
\noindent
Note that $\max v_\mathrm{cond}=\min v_\mathrm{norm}$ as
$\min V_\mathrm{norm}-\max V_\mathrm{cond}=o(N_\mathrm{p})$.

In Ref.~\cite{QPTA} we have proved the following
\begin{theorem}
  \label{theorem1}
  Let $\mathcal{H}$ be the Hilbert space of the system described by
  the Hamiltonian (\ref{H}) with $K$ and $V$ scaling linearly with the
  number of particles $N_\mathrm{p}$.  Let
  $(\max V_\mathrm{cond})_{N_\mathrm{p}}$ be a sequence with rescaled
  per particle $\TDlim $ given by Eq.~(\ref{maxvcond1}).  Set
  $M=\dim \mathcal{H}$ and
  $M_\mathrm{cond}=\dim \mathcal{H}_\mathrm{cond}$, where
  $\mathcal{H}_\mathrm{cond}$ is the span of the $M_\mathrm{cond}$
  eigenvectors of $V$ with eigenvalues at most $\max V_\mathrm{cond}$.
  If
  \begin{align}
    \label{QPT0}
    \TDlim {M_\mathrm{cond}/M}=0,
  \end{align}
  then
  \begin{align}
    \label{QPT1}
    \epsilon = \min\{\epsilon_\mathrm{cond},\epsilon_\mathrm{norm}\},
  \end{align}
  where $\epsilon, \epsilon_\mathrm{cond}, \epsilon_\mathrm{norm}$ are
  the rescaled per particle energies of
  Definition~\ref{RescaledEnergies}.
\end{theorem}
The proof derived in~\cite{QPTA} is based on functional analysis but
we have proved and generalized Theorem~\ref{theorem1} also via a
dynamical approach, and even extended it to finite temperature where
the energies $\epsilon$, $\epsilon_\mathrm{cond}$ and
$\epsilon_\mathrm{norm}$ are replaced by the corresponding free
energies~\cite{FTQPT,FTQPT2}.  Theorem~\ref{theorem1} establishes the
possibility, not yet the realization, of a QPT between a normal phase
characterized by the energy per particle $\epsilon_\mathrm{norm}$,
obtained by removing from $\mathcal{H}$ the infinitely smaller
subspace $\mathcal{H}_\mathrm{cond}$, and a condensed phase
characterized by the energy per particle $\epsilon_\mathrm{cond}$,
obtained by restricting the system to $\mathcal{H}_\mathrm{cond}$.

The situation is particularly simple for systems characterized by a
single Hamiltonian parameter as in the case of Eq.~(\ref{H}).  If
condition~(\ref{QPT0}) holds and, moreover, the functions
$\epsilon_\mathrm{norm}(g)$ and $\epsilon_\mathrm{cond}(g)$ are such
that the equation
\begin{align}
  \label{QCP}
  \epsilon_\mathrm{norm}(g)=\epsilon_\mathrm{cond}(g)
\end{align}
admits a unique \textit{finite} solution $g=g_\mathrm{c}$, then
Eq.~(\ref{QPT1}) provides
\begin{align}
  \label{QPT2}
  \epsilon(g) = \left\{
  \begin{array}{ll}
    \epsilon_\mathrm{norm}(g), \qquad&g<g_\mathrm{c},
    \\
    \epsilon_\mathrm{cond}(g), \qquad&g>g_\mathrm{c}.
  \end{array}
  \right.
\end{align}
Equations~(\ref{QCP}-\ref{QPT2}) imply the existence of a QPT at the
critical point $g_\mathrm{c}$ where the system splits between the
normal ($g<g_\mathrm{c}$) and condensed ($g>g_\mathrm{c}$) phases. In
Ref.~\cite{QPTA} we claimed this QPT to be first-order due to the fact
that, although in general $\epsilon_\mathrm{cond}(g)$ and
$\epsilon_\mathrm{norm}(g)$ are separately analytic also at
$g=g_\mathrm{c}$, they are different functions, so that, while
$\epsilon(g)$ is continuous at $g=g_\mathrm{c}$, its first derivative
was expected to undergo the discontinuity
$|\epsilon_\mathrm{cond}'(g_\mathrm{c})-
\epsilon_\mathrm{norm}'(g_\mathrm{c})|>0$.  This is in fact the case
for the models we focused on in Ref.~\cite{QPTA}, where
$M_{\mathrm{cond}}$ was equal to the degeneracy of $\min V$, i.e.,
$\mathcal{H}_{\mathrm{cond}}$ was defined as the subspace spanned by
the eigenstates of V with minimal eigenvalue.  In these cases, the
energy $\epsilon_\mathrm{cond}(g)$ does not exhibits any kinetic
contribution,
$\langle
E_\mathrm{cond}(N,N_\mathrm{p},g)|K|E_\mathrm{cond}(N,N_\mathrm{p},g)
\rangle =0$, so that $\epsilon_\mathrm{cond}(g)$ and
$\epsilon_\mathrm{norm}(g)$ have different functional dependence on
$g$ and a different first-order derivative at the crossing point
$g_\mathrm{c}$.  However, later, in Refs.~(\cite{WC_QPT,FTQPT,FTQPT2}
we have relaxed such a rigid and limitative definition of
$\mathcal{H}_{\mathrm{cond}}$.  In this more general class of QPTs,
there is no reason to assume that
$|\epsilon_\mathrm{cond}'(g_\mathrm{c})-
\epsilon_\mathrm{norm}'(g_\mathrm{c})|>0$, i.e., the QPT could also be
a second-order transition.  In general, our approach is not able to
detect \textit{a priori} whether the QPT is first- or second-order.
At any rate, in the $\TDlim$, due to Eq. (\ref{Space}), our class of
QPTs always implies the orthogonality between any pair of GSs where
one is in the condensed and the other in the normal phase, namely
\begin{align}
  \label{Ortho}
  \TDlim  ~\langle E (N,N_\mathrm{p},g)|E (N,N_\mathrm{p},g')\rangle =0,
  \quad g<g_\mathrm{c},\quad g'>g_\mathrm{c}.
\end{align}
Eq. (\ref{Ortho}) is fully consistent with the so called ``fidelity
approach'' to QPTs~\cite{Fidelity}.

Additionally, Eq.~(\ref{QCP}) could have more than one solution and
Eq.~(\ref{QPT2}) should properly be rewritten.  For example, if there
exist two critical points, $g_\mathrm{c}$ and $g'_\mathrm{c}$, with
$g_\mathrm{c}<g'_\mathrm{c}$, Eq.~(\ref{QPT2}) must be replaced by
\begin{align}
  \label{QPT2b}
  \epsilon(g) = \left\{
  \begin{array}{ll}
    \epsilon_\mathrm{norm}(g), \qquad& g<g_\mathrm{c},
    \\
    \epsilon_\mathrm{cond}(g), \qquad& g_\mathrm{c}<g<g'_\mathrm{c},
    \\
    \epsilon_\mathrm{norm}(g), \qquad& g'_\mathrm{c}<g.
  \end{array}
  \right.
\end{align}
In the next two sections we shall focus on cases in which
Eq.~(\ref{QCP}) admits at least one solution while later we shall
discuss the general case.

Summarizing, if we find a partition
$\mathcal{H}=\mathcal{H}_\mathrm{cond} \oplus
\mathcal{H}_\mathrm{norm}$ such that {condition}~(\ref{QPT0}) and
Eq.~(\ref{QCP}) are satisfied, then a QPT occurs at some
$g=g_\mathrm{c}$.  In general, such a partition is not unique. In
fact, for Eq.~(\ref{QCP}) to admit a solution with condition
(\ref{QPT0}) satisfied, $\mathcal{H}_\mathrm{cond}$ can invariantly be
chosen provided that it is \textit{not too small} and \textit{not too
  large} in such a way that neither of the two restrictions of $H$, to
$\mathcal{H}_\mathrm{cond}$ and to $\mathcal{H}_\mathrm{norm}$, have a
QPT. In this case, $\epsilon_\mathrm{cond}$ and
$\epsilon_\mathrm{norm}$ are both analytic functions of $g$ at
$g=g_\mathrm{c}$, whereas $\epsilon$ is not.  It must be in fact
considered that, under certain extreme choices of
$\mathcal{H}_\mathrm{cond}$, Eq.~(\ref{QCP}) might have no solution
even if the system owns an actual QPT. This occurs when the chosen
value of $\TDlim \max V_\mathrm{cond}/N_\mathrm{p}$ is either too
small or too large. Note that the latter ``too large'' case can occur
even if it satisfies the necessary condition (\ref{QPT0}). Because
Eq.~(\ref{QCP}) has no solution, this leads to
$\epsilon(g)=\epsilon_\mathrm{norm}(g)$ for any $g$ in the former
case, or to $\epsilon(g)=\epsilon_\mathrm{cond}(g)$ for any $g$ in the
latter. Such extreme choices of
$\TDlim \max V_\mathrm{cond}/N_\mathrm{p}$ prevent the detection of a
putative QPT in terms of a condensation in the space of states and, in
order to avoid them, we introduce the following restrictions on the
possible choices of the condensed space:
\begin{definition}
  \label{upperF}
  Among all possible choices of $\max v_\mathrm{cond}$ such that
  $\TDlim M_\mathrm{cond}/M=0$, we call $\overline{v}_\mathrm{cond}$
  the largest value of $\max v_\mathrm{cond}$ for which
  $\epsilon_\mathrm{norm}(0)=\epsilon_\mathrm{cond}(0)$, i.e.,
  $\epsilon_\mathrm{norm}(0)>\epsilon_\mathrm{cond}(0)$ for
  $\max v_\mathrm{cond}>\overline{v}_\mathrm{cond}$. If instead
  $\epsilon_\mathrm{norm}(0)<\epsilon_\mathrm{cond}(0)$ for any value
  of $\max v_\mathrm{cond}$, we set
  $\overline{v}_\mathrm{cond}=\media{v}$.
\end{definition}
\begin{definition}
  \label{lowerF}
  Among all possible choices of $\max v_\mathrm{cond}$ such that
  $\TDlim M_\mathrm{cond}/M=0$, we call $\underline{v}_\mathrm{cond}$
  the lowest value of $\max v_\mathrm{cond}$, namely,
  $\max v_\mathrm{cond}=\min v\nonumber$.  We also set
  $\underline{k}_\mathrm{cond}= \min k_\mathrm{cond}$.
\end{definition}
In the last line of Definition~\ref{lowerF}, we characterize the
smallest choice of $\mathcal{H}_\mathrm{cond}$ via its kinetic energy.

\section{Existence and location of a QPT}
Whereas Eq.~(\ref{QPT0}) can be checked easily, the existence of a
finite solution to Eq.~(\ref{QCP}) and its location can be made by
using the following approach, firstly developed in Ref.~\cite{WC_QPT}
for a specific Hamiltonian, that here we generalize to any Hamiltonian
$H$ of the form (\ref{H}).

For $N,N_\mathrm{p}$ finite with $N_\mathrm{p}/N=\varrho$ constant, we
evaluate $g_\mathrm{cross}(N,N_\mathrm{p})$ defined as as the value of
the parameter $g$, if any, solution of the equation
\begin{align}
  \label{QCPfinite}
  E_\mathrm{norm}(N,N_\mathrm{p},g)=E_\mathrm{cond}(N,N_\mathrm{p},g).
\end{align}
Assuming a smooth limiting behavior, we expect
\begin{align}
  \label{gc_gcross}
  g_\mathrm{c} = \TDlim g_\mathrm{cross}(N,N_\mathrm{p}).
\end{align}

Note that, for finite sizes, different partitions
$\mathcal{H}=\mathcal{H}_\mathrm{cond} \oplus
\mathcal{H}_\mathrm{norm}$ lead, in general, to different values of
both $E_\mathrm{cond}(g)$ and $E_\mathrm{norm}(g)$.  Only in the
$\TDlim$ different invariant partitions of $\mathcal{H}$ lead to the
same values of $\epsilon_\mathrm{cond}(g)$ for $g>g_\mathrm{c}$ and
$\epsilon_\mathrm{norm}(g)$ for $g<g_\mathrm{c}$, namely, to a unique
$\epsilon(g)$, as indicated by Eq.~(\ref{QPT2}).  In the following, we
shall exploit this invariance, i.e., the invariance with respect to
the choice of $\max v_\mathrm{cond}$, with
$\underline{v}_\mathrm{cond}<\max v_\mathrm{cond}\leq
\overline{v}_\mathrm{cond}$, to prove that a QPT exists and derive
rigorous bounds to $g_\mathrm{c}$.  To this aim, we need first to
provide a few definitions and a lemma.

\begin{definition}
  \label{Selfave}
  Let $P(V)$ be the (probability) distribution of the eigenvalues of
  $V$
  \begin{align}
    \label{Vdistribution}
    P(V) = \frac{1}{M} \sum_{n=1}^M \delta_{V_n,V}. 
  \end{align}
  We say that the potential $V$ is self-averaging if
  \begin{align}
    \label{SelfaveP}
    \TDlim P\left(\frac{V_n}{N_\mathrm{p}}\right) =
    \delta \left(\TDlim\frac{V_n}{N_\mathrm{p}}-
    \TDlim\frac{\media{V}}{N_\mathrm{p}}\right).
  \end{align}
\end{definition}
In some cases, it is possible to prove the self-averaging
property~(\ref{SelfaveP}) analytically as, for example, in a system of
noninteracting fermions subjected to a heterogeneous external field
like $V=\sum^{N/2}_{i=1}c^{\dag}_i c_i$, where $c_i$ stands for the
annihilation operator acting on site $i$, or else in a system of
fermions interacting via a screened Coulomb potential, as proved in
Ref.~\cite{WC_QPT}, but also in spin systems with Ising interactions,
as we shall show in detail later as a benchmark application. In
systems with more complicated potentials, the self-averaging property
can be checked numerically quite easily, as the distribution $P(V)$ of
the potential values can be straightforwardly evaluated for huge
system sizes.  In general, the self-averaging property is expected to
be quite common in physical systems, i.e., in systems where $V$
satisfies the basic requirement that the number of its distinct
eigenvalues grows linearly with $N$ while the state space, i.e., the
total number of eigenvalues of $V$, grows exponentially with $N$, so
that the relative variance of $V$ vanishes in the TD-lim. Note that,
in general, a spin-glass model may not fall into this class.
 
Also the evaluation of the parameters $\overline{v}_\mathrm{cond}$,
$\underline{v}_\mathrm{cond}$, and $\underline{k}_{\mathrm{cond}}$
does not constitute a challenging numerical problem since these
quantities can be derived by a Monte Carlo simulation of the
noninteracting system with $g=0$.

Clearly, if $V$ is self-averaging, condition (\ref{QPT0}) turns out to
be satisfied if, and only if, $\max v_\mathrm{cond}<\media{v}$, which
implies also $\overline{v}_\mathrm{cond}\leq \media{v}$.  On combining
this with definition \ref{upperF}, Theorem \ref{theorem1} via
Eq. (\ref{QPT2}) gives the following
\begin{lemma}
  \label{lemma}
  Let $V$ be self-averaging and set
  $\max v_\mathrm{cond} <\overline{v}_\mathrm{cond}$, then
  $\TDlim M_\mathrm{cond}/M=0$ and
  $\min k_\mathrm{norm}=\min k<\min k_\mathrm{cond}$.
\end{lemma}
We are now ready to state our main result.
\begin{theorem}
  \label{theorem-main}
  Let $K$ and $V$ scale linearly with the number of particles
  $N_\mathrm{p}$ and $V$ be self-averaging. Then, in the $\TDlim$, the
  system with Hamiltonian $H=K+gV$ undergoes at least one QPT at a
  critical point $g=g_\mathrm{c}$ such that
  $g^-_\mathrm{c} \leq g_\mathrm{c} \leq g^+_\mathrm{c}$, where
  \begin{subequations}
    \label{ResultC}
    \begin{align}
      & g^-_\mathrm{c} =
        \frac{{\underline{k}_{\mathrm{cond}}}-\min k}{\max v - \min v}, 
      \\
      \label{ResultD}
      & g^+_\mathrm{c} =
        \frac{-\min k}{{\overline{v}_\mathrm{cond}}-\min v}.
    \end{align}
  \end{subequations}
\end{theorem}
\begin{proof}
  Let us choose an arbitrary partition
  $\mathcal{H}=\mathcal{H}_\mathrm{cond}\oplus
  \mathcal{H}_\mathrm{norm}$ such that
  $\max v_\mathrm{cond}< \overline{v}_\mathrm{cond}$.  By Lemma
  \ref{lemma}, the self-averaging property implies that
  condition~(\ref{QPT0}) is satisfied.
 
  To comply with Eq.~(\ref{QCP}), we exploit the concavity of
  $\epsilon(g)$. The concavity of the GS energy holds for any
  Hamiltonian of the form (\ref{H}) and for any size as can be proven
  from the variational principle
  $E(g)=\min_{\psi\in \mathcal{H}:~\norm{\psi} =1} \langle \psi |H|
  \psi\rangle$ (for shortening the notation here and below we drop the
  dependence on $N$ and $N_\mathrm{p}$ of the GS energy of the system
  of finite size).  In fact, let $a_1$ and $a_2$ be two real non
  negative numbers such that $a_1+a_2=1$.  The linearity of the
  operator $H$ in $g$ and the inequality
  $\min (A + B)\geq \min A + \min B$, valid for any pair of Hermitian
  operators $A$ and $B$, imply
  \begin{align}
    \label{Concavity2}
    E (a_1 g_1+a_2 g_2)
    &= \min_{\psi\in \mathcal{H}:~\norm{\psi} =1}
      \langle \psi |a_1(K+g_1 V)+a_2(K+g_2 V)| \psi\rangle
      \nonumber \\
    & \geq a_1 E (g_1)+ a_2 E (g_2), \qquad \forall ~g_1, ~g_2.
  \end{align}
  An identical formula also applies to the GS energies of the
  Hamiltonian restricted to the subspaces $\mathcal{H}_\mathrm{cond}$
  and $\mathcal{H}_\mathrm{norm}$.  In other words, $E(g)$,
  $E_\mathrm{cond}(g)$, and $E_\mathrm{norm}(g)$ are concave functions
  of $g$ and so are their $\TDlim$ $\epsilon(g)$,
  $\epsilon_\mathrm{cond}(g)$, and $\epsilon_\mathrm{norm}(g)$. Note
  that this fact implies a strong constraint between
  $\epsilon_\mathrm{cond}(g)$ and $\epsilon_\mathrm{norm}(g)$: if
  $g_\mathrm{c}$ is a crossing point such that
  $\epsilon_\mathrm{norm}(g)<\epsilon_\mathrm{cond}(g)$ when
  $g<g_\mathrm{c}$ and
  $\epsilon_\mathrm{cond}(g)<\epsilon_\mathrm{norm}(g)$ when
  $g_\mathrm{c}<g$, for $\epsilon(g)$ to be concave it must be
  $\epsilon'_\mathrm{norm}(g_\mathrm{c})\geq
  \epsilon'_\mathrm{cond}(g_\mathrm{c})$.  This structure, as pictured
  for example in Fig.~\ref{sketch_V_pos},
  implies that there exists a finite solution $g_\mathrm{c}$ of
  Eq.~(\ref{QCP}) if and only if the following two conditions are met:
  (i) $\epsilon_\mathrm{norm}(0) < \epsilon_\mathrm{cond}(0)$ and (ii)
  $\lim_{g\to\infty} \epsilon_\mathrm{cond}(g)/g < \lim_{g\to\infty}
  \epsilon_\mathrm{norm}(g)/g$.
  \begin{figure}[tb]
    \centering
    \includegraphics[width=0.50\columnwidth,clip]{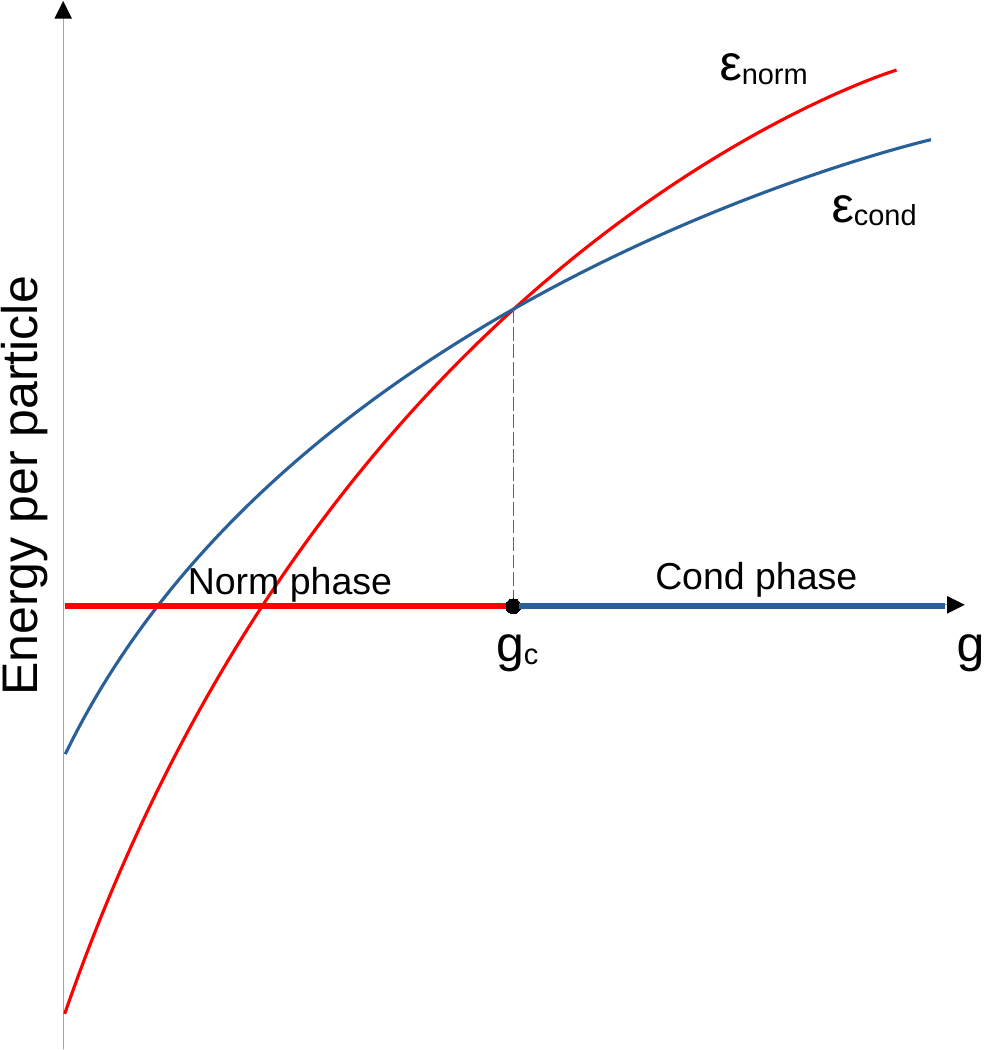}
    \caption {Illustration for Theorem \ref{theorem-main}. The red and
      blue continuous lines are, respectively, the energies per
      particle $\epsilon_\mathrm{norm}(g)$ and
      $\epsilon_\mathrm{cond}(g)$ depicted for $g\geq 0$.  The two
      lines intersect at the finite $g_\mathrm{c}$, and both are
      negative at small $g$ (due to $K$ nonpositive) and increase by
      increasing $g$ (due to $V$ nonnegative).}
    \label{sketch_V_pos}
  \end{figure}

  Condition (i) is equivalent to say that
  $\min k_\mathrm{norm} < \min k_\mathrm{cond}$.  This inequality,
  however, is granted by Lemma~\ref{lemma} whose hypotheses are
  satisfied again choosing
  $\max v_\mathrm{cond} < \overline{v}_\mathrm{cond}$, i.e.,
  $\max v_\mathrm{cond} \leq \overline{v}_\mathrm{cond}-\varepsilon$,
  with $\varepsilon$ positive arbitrarily small.
 
  Condition (ii) is equivalent to say that
  $\min v_\mathrm{cond} < \min v_\mathrm{norm}$.  Since
  $\min v_\mathrm{cond} = \min v$ and
  $\min v_\mathrm{norm} = \max v_\mathrm{cond}$, the latter inequality
  amounts to require $\max v_\mathrm{cond} > \min v$, i.e.,
  $\max v_\mathrm{cond}\geq \underline{v}_\mathrm{cond} +
  \varepsilon'$, with $\varepsilon'$ positive arbitrarily small.
  
  In conclusion, the existence of any partition
  $\mathcal{H}=\mathcal{H}_\mathrm{cond} \oplus
  \mathcal{H}_\mathrm{norm}$, obtained by choosing
  $\underline{v}_\mathrm{cond} + \varepsilon' \leq \max
  v_\mathrm{cond} \leq \overline{v}_\mathrm{cond}-\varepsilon$, with
  $\varepsilon$ and $\varepsilon'$ positive arbitrarily small, allows
  us to claim that Eq.~(\ref{QPT0}) is satisfied and that
  Eq.~(\ref{QCP}) admits at least a finite solution $g_\mathrm{c}$
  where, in the $\TDlim$, the system (\ref{H}) undergoes a QPT.

  The rest of the proof is devoted to the construction of upper and
  lower bounds of $g_\mathrm{c}$. To this end, we shall again exploit
  the invariance of the $\TDlim$~(\ref{gc_gcross}) under different
  partitions
  $\mathcal{H}=\mathcal{H}_\mathrm{cond} \oplus
  \mathcal{H}_\mathrm{norm}$.
 
  At any finite size of the system, since $E_\mathrm{cond}$ and
  $E_\mathrm{norm}$ are concave functions of $g$, we have
  \begin{align}
    \label{gcross_bounds}
    g^-_\mathrm{cross} \leq g_\mathrm{cross} \leq g^+_\mathrm{cross},
  \end{align}
  where $g^+_\mathrm{cross}$ is the intersection point of two curves
  which are, respectively, a majorant of $E_\mathrm{cond}$ and a
  minorant of $E_\mathrm{norm}$, whereas $g^-_\mathrm{cross}$ is the
  intersection point of two curves which are, respectively, a minorant
  of $E_\mathrm{cond}$ and a majorant of $E_\mathrm{norm}$ (see
  Fig.~\ref{sketch_V_neg_with_boundings} for an illustration).
  \begin{figure}
    \centering \includegraphics[width=0.50\columnwidth,clip]%
    {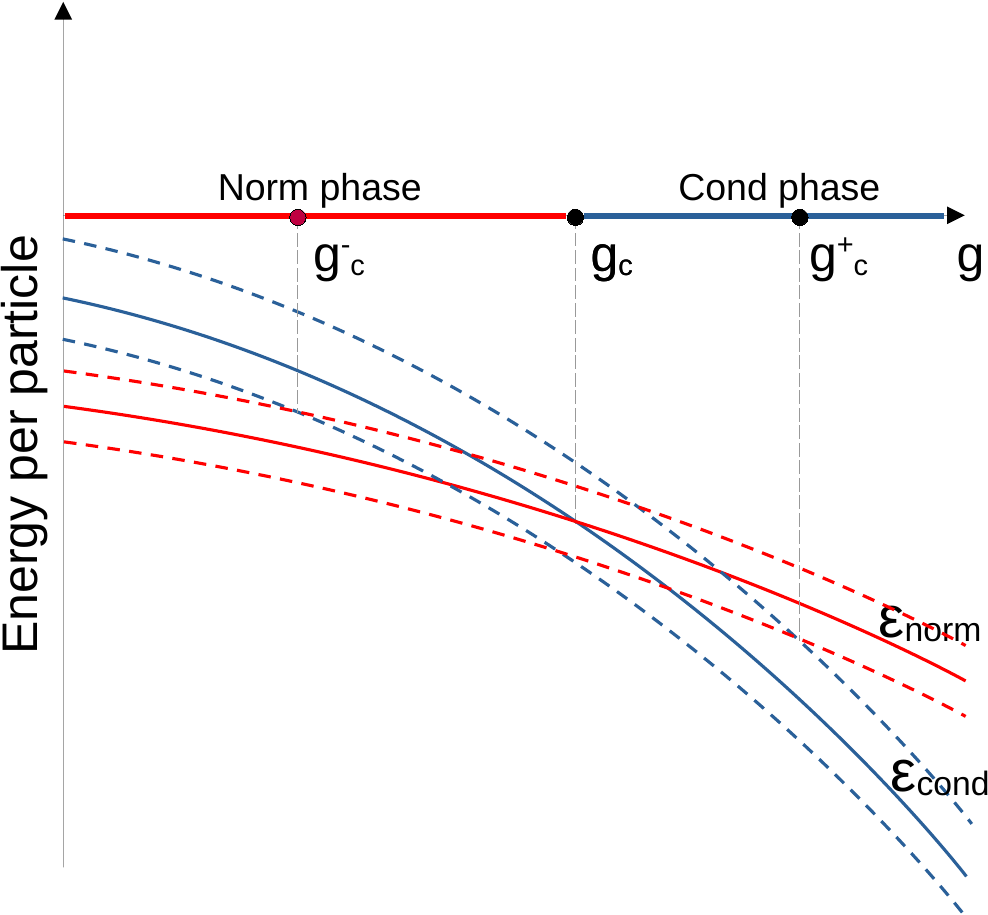}
    \caption {Illustration for the proof Theorem
      \ref{theorem-main}. The red and blue continuous lines are the
      energies per particle $\epsilon_\mathrm{norm}(g)$ and
      $\epsilon_\mathrm{cond}(g)$, respectively, while their close
      dashed lines represents majorants and minorants of them.}
    \label{sketch_V_neg_with_boundings}
  \end{figure}
  Indicating with $g^\pm_\mathrm{c}$ the $\TDlim$s of
  $g^\pm_\mathrm{cross}$, we then have
  $g^-_\mathrm{c} \leq g_\mathrm{c} \leq g^+_\mathrm{c}$. The more
  accurate are the approximations to $E_\mathrm{cond}$ and
  $E_\mathrm{norm}$, the tighter are the bounds
  $g^\pm_\mathrm{c}$. However, we also want to choose these
  approximations to $E_\mathrm{cond}$ and $E_\mathrm{norm}$
  sufficiently simple to allow for an analytical evaluation of the
  $\TDlim$ of $g^\pm_\mathrm{cross}$.

  Let us examine the following inequalities
  \begin{align}
    \label{Econd+}
    E_\mathrm{cond} (g)
    &\leq g \min V_\mathrm{cond},
    \\
    \label{Enorm-}
    E_\mathrm{norm}(g)
    &\geq \min K_\mathrm{norm} + g \min V_\mathrm{norm},
  \end{align}
  and
  \begin{align}
    \label{Econd-}
    E_\mathrm{cond}(g)
    &\geq \min K_\mathrm{cond} + g \min V_\mathrm{cond},
    \\
    \label{Enorm+}
    E_\mathrm{norm}(g)
    &\leq \min K_\mathrm{norm} + g \max V_\mathrm{norm}.
  \end{align}
  Equations~(\ref{Enorm-}), (\ref{Econd-}) and (\ref{Enorm+})
  correspond to Weyl's inequalities~\cite{MatrixTheory} for the lowest
  eigenvalue of $H=K+gV$ restricted to the condensed and normal
  subspaces.  Equation~(\ref{Econd+}) follows from
  $E_\mathrm{cond} \leq \braket{u}{H}{u}/\scp{u}{u}$,
  $\forall\ket{u}\in\mathcal{H}_\mathrm{cond}$, choosing
  $\ket{u}=\ket{n}$, where $\ket{n}$ is any GS of $V$, and observing
  that $\braket{n}{K}{n}=0$.  From the first and second pair of
  inequalities we obtain, respectively,
  \begin{align}
    \label{gcross+}
    g^+_\mathrm{cross}
    = \frac{-\min K_\mathrm{norm}}
    {\min V_\mathrm{norm}-\min V_\mathrm{cond}},
  \end{align}
  \begin{align}
    \label{gcross-}
    g^-_\mathrm{cross}
    = \frac{\min K_\mathrm{cond}-\min K_\mathrm{norm}}
    {\max V_\mathrm{norm}-\min V_\mathrm{cond}}.
  \end{align}

  Consider Eq.~(\ref{gcross+}), we look for a choice of
  $\mathcal{H}_\mathrm{cond}$ which makes the value of this finite
  size majorant of $g_\mathrm{cross}$ as small as possible and then
  take its $\TDlim$.  The numerator is invariant under a change of the
  Hilbert space partition as, according to Lemma \ref{lemma},
  $\TDlim \min K_\mathrm{norm}/N_\mathrm{p} = \TDlim \min K
  /N_\mathrm{p}$.  In the denominator, $\min V_\mathrm{cond} = \min V$
  is invariant too but we can maximize $\min V_\mathrm{norm}$ by
  taking $\mathcal{H}_\mathrm{cond}$ to have the largest possible
  dimension, namely,
  $\TDlim\min V_\mathrm{norm}/N_\mathrm{p} = \TDlim \max
  V_\mathrm{cond}/N_\mathrm{p} \to \overline{v}_\mathrm{cond}$.  In
  the $\TDlim$ we thus obtain the following smallest majorant of
  $g_\mathrm{c}$
  \begin{align}
    \label{gc+}
    g^+_\mathrm{c} = \frac{-\min k}
    {{\overline{v}_\mathrm{cond}}-\min v}.
  \end{align} 

  Consider Eq.~(\ref{gcross-}), we now look for a choice of
  $\mathcal{H}_\mathrm{cond}$ which makes the value of this finite
  size minorant of $g_\mathrm{cross}$ as large as
  possible~\cite{note_independent_choices}.  The denominator is
  invariant under a change of the Hilbert space partition, we have
  $\max V_\mathrm{norm} = \max V$ and $\min V_\mathrm{cond} = \min V$.
  In the numerator, $\min K_\mathrm{norm}$ is invariant too, as Lemma
  \ref{lemma} provides
  $\TDlim \min K_\mathrm{norm}/N_\mathrm{p} = \TDlim \min K
  /N_\mathrm{p}$.  It remains to choose $\mathcal{H}_\mathrm{cond}$ so
  as to maximize $\min K_\mathrm{cond}$. This is obtained by assuming
  for $\mathcal{H}_\mathrm{cond}$ the smallest size, which amounts to
  have
  $\min K_\mathrm{cond}/N_\mathrm{p} \to
  \underline{k}_\mathrm{cond}$. In the $\TDlim$ we thus obtain the
  following largest minorant of $g_\mathrm{c}$
  \begin{align}
    \label{gc-}
    g^-_\mathrm{c} =
    \frac{{\underline{k}_{\mathrm{cond}}}-\min k}{\max v - \min v}.
  \end{align} 
  Note that, in general, $\underline{k}_\mathrm{cond}$ might be
  finite. Moreover, according to the above condition (i),
  $\underline{k}_\mathrm{cond} = \epsilon_\mathrm{cond}(0) >
  \epsilon_\mathrm{norm}(0) = \min k$ so that $g^-_\mathrm{c}>0$.
\end{proof}

A few major comments are in order with respect to what stated in
Ref.~\cite{WC_QPT}.
  
i) In present Theorem \ref{theorem-main}, at no point we have assumed
the functions $\epsilon(g)$, $\epsilon_\mathrm{cond}(g)$, and
$\epsilon_\mathrm{norm}(g)$ to be monotone; only their concavity
matters~\cite{note_monotone}.
  
ii) In present Theorem \ref{theorem-main} there is no restriction on
the nature of the QPT, being possible both first- as well as
second-order QPTs; in the latter case, the critical point is such that
$\epsilon_\mathrm{norm}(g)$ and $\epsilon_\mathrm{cond}(g)$ cross each
other but share the same tangent~\cite{note_example}.

iii) For all models where $\min V$ has degeneracy at most
$\mathop{o}(N_\mathrm{p})$ (as for the model considered in
Ref.~\cite{WC_QPT}), we have $\underline{k}_{\mathrm{cond}}=0$;
however, there exist models where this is not the case, as for the
model of a system of fermions in a one-dimensional inhomogeneous
lattice analyzed in Refs.~\cite{QPTA}.

iv) In Ref.~\cite{WC_QPT}, we have arbitrarily assumed
$\overline{v}_\mathrm{cond}=\media{v}$, which, in general, may not be
correct.  However, by using
$\overline{v}_\mathrm{cond}\leq \media{v} \leq \max v$, the following
result renders the assumption of Ref.~\cite{WC_QPT} effective.
\begin{corollary}
  \label{corollary-main}
  Under the hypotheses of theorem~\ref{theorem-main}
  \begin{align}
    \label{corollaryEq1}
    g^+_\mathrm{c} \geq \tilde{g}^+_\mathrm{c} \geq
    \tilde{g}^-_\mathrm{c} \geq g^-_\mathrm{c},
  \end{align}
  where
  \begin{align}
    \label{corollaryEq2}
    &\tilde{g}^+_\mathrm{c} = \frac{-\min k}{{\media{v}}-\min v},
    \\
    \label{corollaryEq3}
    &\tilde{g}^-_\mathrm{c}= \frac{\underline{k}_\mathrm{cond}-\min k}
      {\media{v} - \min v}. 
  \end{align}
\end{corollary}
For systems where $\underline{k}_\mathrm{cond}=0$, as in the model of
Ref.~\cite{WC_QPT}, as well as in any model where $\min V$ has
degeneracy at most $\mathop{o}(N_\mathrm{p})$, the quantities
$\tilde{g}^-_\mathrm{c}$ and $\tilde{g}^+_\mathrm{c}$ assume the same
value, which suggests the following approximated formula for the
critical point
\begin{align}
  \label{gcapprox}
  g_\mathrm{c} \approx \frac{-\min k}{{\media{v}}-\min v}. 
\end{align}
While the accuracy of this formula will be discussed in the examples
illustrated in the next Section, we stress the importance of its
simplicity. To estimate $g_\mathrm{c}$, it does not require the
evaluation of $\overline{v}_\mathrm{cond}$, which, in general, remains
analytically impossible and numerically non trivial.

\section{Applications}
Below we provide a few examples to show how Theorem \ref{theorem-main}
and Corollary \ref{corollary-main} apply.  We anticipate that all
these models have self-averaging potentials $V$ (as mentioned earlier
this is in fact a quite common feature of extensive Hamiltonian
systems), which, according to the first part of Theorem
\ref{theorem-main}, implies the existence of at least a QPT (first- or
second-order). Moreover, the potentials $V$ of these examples have
trivial minimal and maximal eigenvalues, allowing for an analytical
application of Eqs.~(\ref{ResultC}) and Eq.~(\ref{gcapprox}).

\subsection{Grover model: an exact case}
Consider the spin Hamiltonian $H=K+gV$, called Grover's
model~\cite{Grover,Bennet,QPTA}, where
\begin{align}
  \label{GroverH}
  K=-\sum_{i=1}^N\sigma_i^x, \qquad V=-N|n_0\rangle\langle n_0|.
\end{align}
Here, $\sigma_i^x$ and $\sigma_i^z$ are Pauli matrices.  In this
model, in which $N_\mathrm{p}=N$, we have $M=2^N$ and
$\mathcal{H}=\Span \{\ket{s_1} \otimes \dots \otimes\ket{s_N}\}$,
where $|s_i\rangle$, with $s_i=\pm 1$, are the eigenstates of the
Pauli matrix $\sigma_i^z$ and $V_n=0$ for any configuration
$\ket{n}=\ket{s_1} \otimes \dots \otimes\ket{s_N}$ except for the
configuration $\ket{n_0}=\ket{-1}^{\otimes N}$, where $V_{n_0}=-N$.
In the thermodynamic limit, this model can be exactly
solved~\cite{QPTA}; it is however instructive to see how Theorem
\ref{theorem-main} and Corollary \ref{corollary-main} apply.  In fact,
this model provides a rare case in which even
$\overline{v}_\mathrm{cond}$ can be calculated analytically.

Given $N$, we have only two possible choices of
$\max V_\mathrm{cond}$: either $\max V_\mathrm{cond}=-N $ or
$\max V_\mathrm{cond}=0$; and it is immediate to check that $V$ is
self-averaging with $\media{v}=0$. However, only the former choice,
i.e., $\mathcal{H}_\mathrm{cond}=\Span\{\ket{n_0}\}$, is compatible
with the condition $\max v_\mathrm{cond}<\media{v}$.  For it, we have
$\min k_\mathrm{norm}=-1$ and $\min k_\mathrm{cond}=0$, namely,
$\epsilon_{\mathrm{norm}}(0)<\epsilon_{\mathrm{cond}}(0)$.  Hence,
according to Definitions \ref{upperF} and \ref{lowerF}, here we are in
a situation where $\overline{v}_\mathrm{cond}=\media{v}$ and
$\underline{k}_\mathrm{cond}=0$.  Finally, by using also $\min v=-1$,
$\max v=0$, and $\min k=-1$, Eqs.~(\ref{ResultC}) provide the exact
result $g_\mathrm{c}^{-}=g_\mathrm{c}^{+}=g_\mathrm{c}=1$,
consistently with Ref.~\cite{QPTA}, where the analysis was made by the
explicit calculation of $\epsilon_{\mathrm{cond}}$ and
$\epsilon_{\mathrm{norm}}$, which also showed that the QPT is of
first-order.  Alternatively, for $g_\mathrm{c}$ we could directly use
the formula (\ref{gcapprox}) which here also coincides with the exact
result.

\subsection{Ising model in a transverse field: an example of a
  second-order QPT}
Let us consider the 1D Ising model ($N_\mathrm{p}=N$) with a
transverse field and open boundary conditions (OBC).  The operators
$K$ and $V$ in the dimensionless Hamiltonian $H=K+gV$ are
\begin{align}
  \label{Ising}
  K=-\sum_{i=1}^N\sigma_i^x, \qquad
  V=-\sum_{i=1}^{N}\sigma_i^z\sigma_{i+1}^z,
\end{align}
where $\sigma_i^x$ and $\sigma_i^z$ are Pauli matrices.  As well
known, this model is exactly solvable~\cite{Pfeuty} with
\begin{align}
  \label{Ising1}
  \epsilon(g)=-\frac{1}{2\pi}\int_{-\pi}^\pi dq
  \left[1+2g\cos(q)+g^2\right]^{\frac{1}{2}},
\end{align}
which implies a singularity of second-order at the critical point
$g_\mathrm{c}=1$.  The eigenvalues of $V$ have not fixed sign. For
example, for the two-fold degenerate GS of $V$, i.e., the two states
with all the spins parallel along the $z$-direction, we have
$\min V=-(N-1)$, while for the two most excited states made by
alternated spin directions with have $\max V=(N-1)$.  Starting with
the two GS states of $V$ having all spins parallel, we can achieve the
$2(N-1)$ degenerate first excited eigenstates of $V$ by reversing the
direction of all the spins having index $i=2,\ldots,N$, or by
reversing the direction of all the spins having index $i=3,\ldots,N$,
and so on.  For all these $2(N-1)$ states $|n\rangle$, obtained by the
insertion of a ``cut'' between contiguous parallel spins between the
sites $i-1$ and $i$, we have $V_n=-(N-1-2)$.  Similarly, by inserting
two cuts, which can be done in $2(N-1)(N-2)/2!$ ways, we have
$V_n=-(N-1-4)$; and so on.  We can therefore express the distinct
eigenvalues of $V$, $V^{(1)}\leq V^{(2)} \leq \ldots $, in terms of
the number of cuts $k=0,1,\ldots,N-1$ as follows
\begin{align}
  \label{Ising3}
  V^{(k)}=-(N-1-2k), \qquad \mathrm{deg}(k)=2\binom{N-1}{k},
\end{align}
where $\mathrm{deg}(k)$ stands for degeneracy of the eigenvalue
$V^{(k)}$~\cite{Note3}.  Equation (\ref{Ising3}) provides the
distribution of the eigenvalues of $V$ as
\begin{align}
  \label{Ising4}
  P\left(V^{(k)}\right)=\frac{1}{2^{N-1}} \binom{N-1}{k},
  \qquad k \in\{0,\ldots,N-1\}.
\end{align}
For $N$ large, by using $N-1 \simeq N$ and the Stirling approximation,
we get
\begin{align}
  \label{Ising5}
  P\left(V^{(k)}\right)\simeq \frac{1}{2^N}\sqrt{\frac{N}{2\pi(N-k) k}}
  ~e^{-N\left[x\log (x)+(1-x)\log(1-x)\right]}, \qquad x=k/N.
\end{align}
The above equation shows that the distribution of the rescaled
eigenvalues $V^{(k)}/N$ satisfies the self-averaging property. More
precisely, the most probable rescaled eigenvalue occurs for
$x=k/N=1/2$, i.e., $V^{(k)}/N=-(N-2k)/N=0$, where
$P\left(V^{(k)}\right)=\sqrt{2/(\pi N)}$, while any other eigenvalue
is exponentially suppressed. In other words, $V$ is self-averaging
with $\TDlim \media{V}/N=0$ and Theorem \ref{theorem-main} guarantees
the existence of at least one QPT.  Summarizing, for the Ising model
we have $\min v=-1$, $\max v=1$, and $\media{v}=0$, while the
evaluation of $\overline{v}_\mathrm{cond}$ remains non trivial.
However, since we also have $\underline{k}_\mathrm{cond}=0$, we can
apply Eq.~(\ref{gcapprox}) which provides $g_\mathrm{c}\approx 1$,
consistently with the exact result $g_\mathrm{c}= 1$~\cite{Pfeuty}.
We can also evaluate the exact lower bound (\ref{gc-}) which gives
$g^-_\mathrm{c}= 1/2$. %
By using
$\overline{v}_\mathrm{cond} \approx (\min v + \media{v})/2=-1/2$, we
can approximate the exact upper bound (\ref{gc+}) as
$g^+_\mathrm{c} \approx 2$, consistent with $g_\mathrm{c} = 1$.
Finally, since we know that the model undergoes only one QPT and that
this is of second-order, this example confirms that --- as anticipated
--- Theorem \ref{theorem-main} covers also second-order QPTs.

\subsection{Fermions in a heterogeneous external field: an example
  with $\underline{k}_\mathrm{cond}\neq 0$}
Consider a system of $N_\mathrm{p}$ electrons undergoing first
neighbor hopping in a 1D lattice ring of $N$ sites with $N$ even.  The
fermions are subject to an attractive heterogeneous external field
acting only on the sites $i=1,\ldots,N/2$.  The Hamiltonian $H=K+gV$
describing this tight-binding model has therefore
\begin{align}
  \label{KFermions}
  K = -\sum_{i=1}^{N}\left(c^\dag_{i}c_{i+1}+c^\dag_{i+1}c_{i}\right),
  \qquad
  V = - \sum_{i=1}^{N/2}c^\dag_i c_i. 
\end{align}
where $c_i$, $i=1,\ldots,N$, are the fermionic annihilation operators
and obey periodic boundary conditions (PBC) $c_{i+N}=c_i$.  Similarly
to what we have seen in the case of the Ising model, it is possible to
show that $V$ is self-averaging, so that a QPT must exist.  It turns
out that it is of first-order~\cite{QPTA}.

Let us first focus on the half-filling case $N_\mathrm{p}/N=1/2$ (in
this case, the distribution of $V_n/N$ for finite $N$ is essentially
equal to that of the Ising model).  We have $\min v=-1$, $\max v=0$,
$\media{v}=-1/2$, and
$\min K = -2 \sin(\pi N_\mathrm{p}/N) / \sin(\pi/N)$, i.e.,
$\min k= -4/\pi$, while the evaluation of $\overline{v}_\mathrm{cond}$
remains non trivial.  However, since we also have
$\underline{k}_\mathrm{cond}=0$, we can apply Eq. (\ref{gcapprox})
which provides $g_\mathrm{c}\approx 8/\pi \simeq 2.55$, to be compared
with the result derived from an exact numerical analysis~\cite{QPTA},
$g_\mathrm{c}\simeq 4.0$.  The discrepancy between the two values is
certainly important but not dramatic, especially when considering the
simplicity of Eq. (\ref{gcapprox}).  We can also evaluate the exact
lower bound (\ref{gc-}) which gives
$g^-_\mathrm{c}= 4/\pi\simeq 1.28$. %
By using
$\overline{v}_\mathrm{cond} \approx (\min v + \media{v})/2=-3/4$, we
can approximate the exact upper bound (\ref{gc+}) as
$g^+_\mathrm{c} \approx 16/\pi\simeq 5.09$, consistent with
$g_\mathrm{c}\simeq 4.0$.

Finally, let us consider the case at filling $N_\mathrm{p}/N=1/4$. We
have $\min v=-1$, $\max v=0$, $\media{v}=-1/2$, and
$\min k= -4\sqrt{2}/\pi$. Here we cannot apply Eq. (\ref{gcapprox})
since $\underline{k}_\mathrm{cond}\neq 0$. In fact, the latter can be
determined by
$\min K = -2 \sin(\pi N_\mathrm{p}/N') /
\sin(\pi/N')+\mathop{O}(1)$~\cite{NoteBC}, with $N'=N/2$, resulting in
$\underline{k}_\mathrm{cond}= -4/\pi$.  On plugging these values into
the exact lower bound (\ref{gc-}) we obtain
$g^-_\mathrm{c}= 4(\sqrt{2}-1)/\pi$. As for the upper bound, by using
$\overline{v}_\mathrm{cond} \approx (\min v + \media{v})/2=-3/2$, we
can approximate the exact upper bound (\ref{gc+}) as
$g^+_\mathrm{c} \approx 16/\pi$, while Eq.~(\ref{corollaryEq2})
provides $\tilde{g}^+_\mathrm{c}= 8\sqrt{2}/\pi$.  The found values,
namely, $g^-_\mathrm{c}\simeq 0.53$ and $g^+_\mathrm{c} \simeq 5.09$
or $\tilde{g}^+_\mathrm{c}\simeq 3.60$, are bounds consistent with the
exact numerical result $g_\mathrm{c}\simeq 2.0$.

\section{On the number and nature of critical points}
Theorem \ref{theorem-main} does not specify the number of critical
points. The hypotheses of Theorem \ref{theorem-main} only imply that
the number of critical points is finite due to the asymptotic
linearity of $\epsilon_\mathrm{cond}(g)$ and
$\epsilon_\mathrm{norm}(g)$.  This can be easily seen as follows. For
$g\to\infty$ we have
$\epsilon_\mathrm{cond}(g)/g\to \TDlim \min
V_\mathrm{cond}/N_\mathrm{p}$ and similarly
$\epsilon_\mathrm{norm}(g)/g\to \TDlim \min
V_\mathrm{norm}/N_\mathrm{p}$, with
$\min v_\mathrm{cond} < \min v_\mathrm{norm}$. In other words,
$\epsilon_\mathrm{cond}(g)$ and $\epsilon_\mathrm{norm}(g)$ tend
asymptotically to two straight lines with different slopes.  This
implies that there exists a finite value $g_P$ above which
$\epsilon_\mathrm{cond}(g)$ and $\epsilon_\mathrm{norm}(g)$ cannot
cross each other. On the other hand, on a finite interval $[0,g_P]$,
two distinct analytic functions can cross each other at most a finite
number of times.

Theorem \ref{theorem-main} does not even specify the nature of a
critical point, that is, whether it is first- or second-order.  In
general, without any further specific information about the operators
$K$ and $V$, all kind of scenarios are possibles, with multiple
critical points of arbitrary nature.  In particular, although a
second-order critical point $g_\mathrm{c}$ requires the double
condition
$\epsilon_\mathrm{cond}(g_\mathrm{c})=\epsilon_\mathrm{norm}(g_\mathrm{c})$
and
$\epsilon'_\mathrm{cond}(g_\mathrm{c})=\epsilon'_\mathrm{norm}(g_\mathrm{c})$,
it is possible to construct specific model examples owning any number
of second-order critical points, see ~\ref{multiple}.

For illustrative purposes, in Figs.~\ref{sketch_V_gen_2} and
\ref{sketch_V_gen_4} we report two different scenarios for a case with
two critical points.

\begin{figure}[b]
  \centering
  \includegraphics[width=0.50\columnwidth,clip]{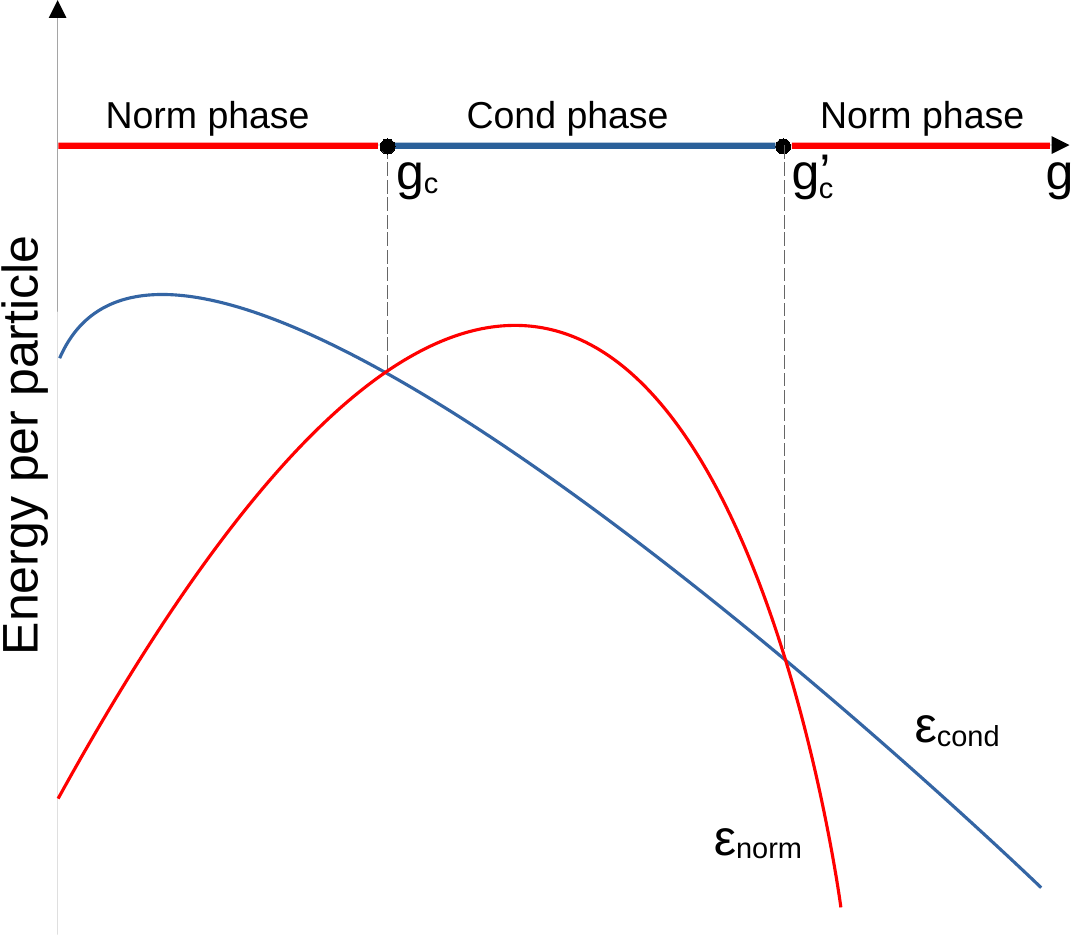}
  \caption {Sketch scenario of Theorem \ref{theorem-main} with two
    first-order QPTs taking place at the critical points
    $g_\mathrm{c}$ and $g'_\mathrm{c}$.  }
  \label{sketch_V_gen_2}
\end{figure}

 \begin{figure}[b]
   \centering
   \includegraphics[width=0.50\columnwidth,clip]{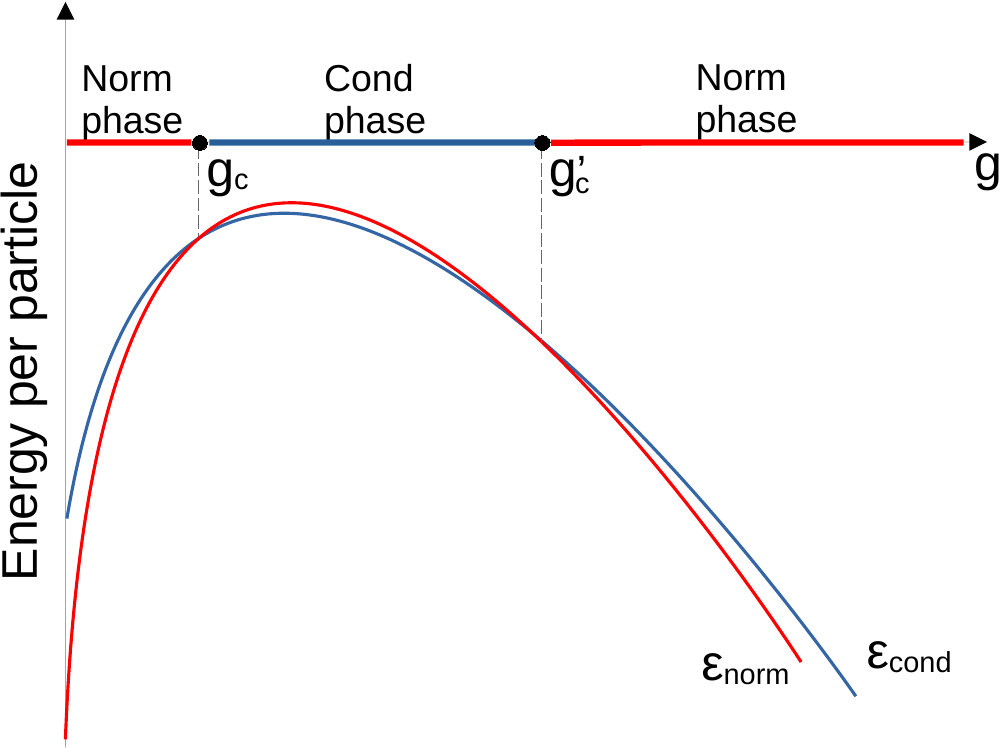}
   \caption {Sketch scenario of Theorem \ref{theorem-main} with a
     first-order QPT and a second-order QPT taking place at the
     critical points $g_\mathrm{c}$ and $g'_\mathrm{c}$, respectively.
   }
   \label{sketch_V_gen_4}
 \end{figure}

 \section{Conclusions}
 \label{conclusions}
 Our analysis suggests that, in lattice systems governed by a
 Hamiltonian $H=K+gV$ (where $K$ and $V$ are generic non-commuting
 operators that scale with system size in the thermodynamic limit),
 any QPT driven by the dimensionless parameter $g$ --- be it first- or
 second-order --- can be interpreted as a condensation in state
 space. In fact, the key property underlying this general behavior is
 the concavity of the ground-state energies of $H$ and its
 restrictions to the norm and cond subspaces. As shown in the proof of
 Theorem \ref{theorem-main} this property, when equipped with the
 self-averaging character of $V$, which is the case for most physical
 systems, leads to the existence of at least one critical point bound
 by quite simple formulas expressed in terms of a few parameters.

 \ack The National Council for Scientific and Technological
 Development (CNPq) and the Foundation for Research Support of the
 State of Bahia (FAPESB) are acknowledged for the productivity
 fellowship no. 302787/2025-9.
 
 %

 \appendix
 \section{Examples with multiple second-order critical points}
 \label{multiple}
 Below we report a step-by-step construction of two analytic, strictly
 concave functions $f(x)$ and $g(x)$ that share a second-order
 crossing at two points, for simplicity $x=0$ and $x=1$, yet have
 different asymptotic slopes $m_f$ and $m_g$ as $x \to \infty$.  The
 construction develops along three steps. At step 1, we define the
 difference function $h(x) = f(x) - g(x)$, which, due to the fact that
 $f$ and $g$ cross each other and share the same tangent at $x_1$ and
 $x_2$, must satisfy $h(x_1)=h(x_2)=0$ and $h'(x_1)=h'(x_2)=0$;
 moreover, we require that $h'(x)\to \Delta_m \neq 0$ as $x\to\infty$,
 where $\Delta_m=m_f-m_g$.  At step 2, by using $h$, we construct the
 individual functions $f$ and $g$ by choosing a suitable negative
 second derivative of $g$ such that $|g''|<|h''|$ which leads to the
 concavity of $f$ too. Finally, at step 3 we check that the above
 hypothesis are all satisfied and conclude.

 \subsection*{Step 1: Construct the difference function
   $h(x) = f(x) - g(x)$.}
 We need an analytic function $h(x)$ that has double roots at $x=0$
 and $x=1$ (meaning $h=0$ and $h'=0$ at these points), and whose slope
 approaches a non-zero constant $\Delta_m$ as $x \to \infty$.

 Let us define the second derivative of $h(x)$ as the everywhere
 analytic expression
 \begin{align*}
   h''(x) = (ax^2 + bx + c)e^{-x^2},
 \end{align*}
 where $a, b,$ and $c$ are constants.  We also define
 $h'(x) = \int_0^x h''(t)\,dt$ and $h(x) = \int_0^x h'(t)\,dt$.

 By the above definition, we automatically satisfy the first crossing
 point: $h(0) = 0$ and $h'(0) = 0$.  To force the second crossing
 point at $x = 1$, we require
 \begin{align*}
   &h'(1) = \int_0^1 (at^2 + bt + c)e^{-t^2} dt = 0,
   \\
   &h(1) = \int_0^1 ~
     \int_0^x (at^2 + bt + c)e^{-t^2} dt ~ \,dx = 0.
 \end{align*}
 These are two simple linear constraints on our three parameters
 $a, b,$ and $c$. We have a 3-dimensional parameter space and only 2
 constraints, meaning there are infinitely many non-trivial solutions
 where $a, b,$ and $c$ are not all zero.

 By choosing one such non-trivial solution, we get an analytic
 function $h(x)$ with double roots at $x=0$ and $x=1$. Furthermore, as
 $x \to \infty$,
 $h'(x) \to \int_0^\infty (at^2 + bt + c)e^{-t^2} dt = \Delta_m$.  By
 carefully picking from our infinite solutions, we can easily ensure
 $\Delta_m \neq 0$. Because $h''(x)$ decays exponentially, $h(x)$
 rapidly approaches a straight line $y = \Delta_m ~ x + q$.

 \subsection*{Step 2: Construct the individual concave functions
   $f(x)$ and $g(x)$.}
 We now need to split $h(x)$ into $f(x) - g(x)$ such that
 $f''(x) \le 0$ and $g''(x) \le 0$ everywhere.

 Because $h''(x) = (ax^2 + bx + c)e^{-x^2}$ decays exponentially, its
 magnitude is strictly bounded. We can find a sufficiently large
 positive constant $K$ such that for all $x$
 \begin{align*}
   |h''(x)| < K e^{-x^2/2}.
 \end{align*}
 Now, define the second derivative of $g(x)$ as
 \begin{align*}
   g''(x) = -K e^{-x^2/2}.
 \end{align*}
 Since $-K e^{-x^2/2} < 0$ everywhere, $g(x)$ (obtained by double
 integration) is strictly concave and analytic.  Finally, define
 $f(x) = g(x) + h(x)$. Its second derivative is
 \begin{align*}
   f''(x) = g''(x) + h''(x) = -K e^{-x^2/2} + h''(x).
 \end{align*}
 Because we chose $K$ to be large enough to dominate the maximum
 positive swing of $h''(x)$, we are guaranteed that $f''(x) < 0$
 everywhere. Thus, $f(x)$ is also strictly concave and analytic.

 \subsection*{Step 3: Verification of asymptotes.}
 The two functions approach straight lines, in fact as $x \to \infty$
 we have
 \begin{align*}
   &g'(x) = \int_0^x -K e^{-t^2/2} dt \to -K \sqrt{\pi/2} = m_g,
   \\
   &f'(x) = g'(x) + h'(x) \to m_g + \Delta_m = m_f.
 \end{align*}
 Because both $f''(x)$ and $g''(x)$ decay exponentially, their
 integrals converge absolutely, meaning the functions $f(x)$ and
 $g(x)$ do not just achieve constant slopes, but they geometrically
 converge to straight-line asymptotes $y = m_f x + c_f$ and
 $y = m_g x + c_g$, respectively.  Since $\Delta m \neq 0$, their
 asymptotic slopes are strictly different ($m_f \neq m_g$).

 In conclusion, $f(x)$ and $g(x)$ are both analytic, globally concave,
 tend asymptotically to straight lines with different slopes, and
 cross exactly twice at $x=0$ and $x=1$ with shared tangents.

 The above scheme for two second-order crossing points can be easily
 generalized to any number of second-order crossing points.  For
 example, for a case with three second-order crossing points at the
 values $x=0$, $x=1$, and $x=2$, the function $h''(x)$ will be defined
 as a Gaussian centered in $x=0$ times a polynomial of fourth degree.
 Four of the five coefficients of this polynomial will be used
 to satisfy the second-order crossings at the points $x=1$ and $x=2$,
 leaving a free coefficient to be chosen as to ensure
 $\Delta_m\neq 0$.  The rest of the construction repeats as in the
 above case with two second-order crossing points.


\vspace{0.2cm}
    
 \begin{thebibliography}{99}%

 \bibitem{SGCS} S.~L.~Sondhi, S.~M.~Girvin, J.~P.~Carini, and
   D.~Shahar, \textit{Continuous quantum phase transitions},
   Rev. Mod. Phys. \textbf{69}, 315 (1997).
  
 \bibitem{KB} T.~R.~Kirkpatrick and D.~Belitz, \textit{Quantum phase
     transitions in electronic systems}, in \textit{Electron
     Correlations in the Solid State} ed. by N.~H.~March, (Imperial
   College Press, London 1999).

 \bibitem{Vojta} T.~Vojta, \textit{Quantum phase transitions in
     electronic systems}, Ann. Phys. (Leipzig) \textbf{9}, 403 (2000).

 \bibitem{Sachdev} S.~Sachdev, \textit{Quantum Phase Transitions}
   (Cambridge University Press, Cambridge 2000).

 \bibitem{Fidelity} H.~T.~Quan and F.~M.~Cucchietti, \textit{Quantum
     fidelity and thermal phase transitions}, Phys. Rev. E
   \textbf{79}, 031101 (2009).
  
 \bibitem{Carr} L.~D.~Carr, \textit{Understanding Quantum Phase
     Transitions}, (CRC Press, Taylor \& Francis 2010).

 \bibitem{LeeYang} R. A. Blythe and M. R. Evans, \textit{The Lee-Yang
     theory of equilibrium and nonequilibrium phase transitions},
   Braz. J. Phys. \textbf{33}, 464 (2003).
  
 \bibitem{QPTA} M.~Ostilli and C.~Presilla, \textit{First-order
     quantum phase transitions as condensations in the space of
     states}, J. Phys. A: Math. Theor. \textbf{54}, 055005 (2021).

 \bibitem{WC_QPT} M.~Ostilli and C.~Presilla, \textit{Wigner
     crystallization of electrons in a one-dimensional lattice: a
     condensation in the space of states},
   Phys. Rev. Lett. \textbf{127}, 040601 (2021).

 \bibitem{FTQPT} M. Ostilli and C. Presilla,
   \textit{Finite-temperature quantum condensations in the space of
     states: general proof}, J. Phys. A: Math. Theor. \textbf{55}
   505004 (2022).
  
 \bibitem{FTQPT2} M. Ostilli and C. Presilla,
   \textit{Finite-temperature quantum condensations in the space of
     states: A different perspective on quantum annealing},
   Phys. Rev. A \textbf{108}, 022205 (2023).

 \bibitem{Wigner} E.~Wigner, \textit{On the Interaction of Electrons
     in Metals}, Phys. Rev. \textbf{46}, 1002 (1934).

 \bibitem{Hubbard1978} J.~Hubbard, \textit{Generalized Wigner lattices
     in one dimension and some applications to
     tetracyanoquinodimethane(TCNQ) salts}, Phys. Rev. B \textbf{17},
   494 (1978).
  

 \bibitem{Note} The adjective ``classical'' here refers to the fact
   that we are considering the mean of the eigenvalues of $V$ while
   any quantum mechanical expectation of $V$ depends also on the state
   in which it is taken.

 \bibitem{EPR} M. Beccaria, C. Presilla, G. F. De Angelis, G. Jona
   Lasinio, \textit{An exact representation of the fermion dynamics in
     terms of Poisson processes and its connection with Monte Carlo
     algorithms}, Europhys. Lett. \textbf{48}, 243 (1999).

 \bibitem{ANAL} M. Ostilli and C. Presilla, \textit{An analytical
     probabilistic approach to the ground state of lattice quantum
     systems: exact results in terms of a cumulant expansion}, J.
   Stat. Mech., P04007 (2005).

 \bibitem{ANAL2} M. Ostilli and C. Presilla, \textit{The Exact ground
     state for a class of matrix Hamiltonian models: quantum phase
     transition and universality in the thermodynamic limit},
   J. Stat. Mech., P11012 (2006).
  
 \bibitem{MatrixTheory} J.~N.~Franklin, \textit{Matrix Theory}, (Dover
   Publications, 1993).

 \bibitem{note_independent_choices} Note that the choices of
   $\mathcal{H}_\mathrm{cond}$ for Eqs.~(\ref{gcross+}) and
   ~(\ref{gcross-}) are independent.
    
 \bibitem{note_monotone} A sufficient condition for $\epsilon(g)$,
   $\epsilon_\mathrm{cond}(g)$, and $\epsilon_\mathrm{norm}(g)$ to be
   monotone non decreasing (non increasing) is that $V_n\geq 0$
   ($V_n\leq 0$) for any $n$, as can be seen using the Hellman-Feynman
   theorem.

 \bibitem{note_example} It is easy to construct cases where two
   concave analytic functions cross at some point where they also
   share same tangent.  Consider for example the functions
   $\epsilon_\mathrm{cond}(g)=\log(1+g) +a g$ and
   $\epsilon_\mathrm{norm}(g)=\epsilon_\mathrm{cond}(g)-\delta(g-1)^3e^{-g}$,
   where $a$ is arbitrary and $\delta$ sufficiently small.  It is then
   easy to check that, for $\delta<1/4$, $\epsilon_\mathrm{cond}(g)$
   and $\epsilon_\mathrm{norm}(g)$ are concave and that
   $\epsilon_\mathrm{cond}(1)=\epsilon_\mathrm{norm}(1)$ as well as
   $\epsilon'_\mathrm{cond}(1)=\epsilon'_\mathrm{norm}(1)=a+1/2$.
  
 \bibitem{Grover} L.~K.~Grover: {\em A fast quantum mechanical
     algorithm for database search\/}, arXiv:quant-ph/9605043.
  
 \bibitem{Bennet} C.~H.~Bennett, E.~Bernstein, G.~Brassard,
   U.~Vazirani, {\em The strengths and weaknesses of quantum
     computation\/}, SIAM Journal on Computing, \textbf{26} 1510
   (1997).

 \bibitem{Pfeuty} P.~Pfeuty, {\em The one-dimensional Ising model with
     a transverse field\/}, Ann. Phys. (N.Y.) \textbf{57}, 79 (1970).

 \bibitem{Note3} Interestingly, also the eigenvalues of the kinetic
   operator $K$ in Eq. (\ref{Ising}) have the same structure as in
   Eq. (\ref{Ising3}) provided $N-1$ is replaced by $N$, and $k$,
   taking the values $0,1,\ldots N$, is interpreted as the number of
   up spins in the eigenbasis of the Pauli operator $\sigma^x_i$.

 \bibitem{NoteBC} The $\mathop{O}(1)$ correction is due to the fact
   that $\underline{k}_\mathrm{cond}$ represents the GS energy of the
   operator $K$ restricted to the system within the first $N'=N/2$
   sites and with OBC, not PBC.
  
 \bibitem{note_conics} For example, one can take a circle and an
   ellipse of suitable sizes and positions such that they cross at
   three points and such that their local convex profile (convex
   hull), including the three crossing points, define a concave
   function.
    
 \end{thebibliography}
\end{document}